\documentclass[11pt]{article}
\usepackage[english]{babel}
\usepackage{amsmath,amsthm}
\usepackage{amsfonts}
\usepackage[comma,numbers,square,sort&compress]{natbib}
\usepackage{graphicx,subfigure,amsmath,amssymb,mathrsfs,amsfonts,amstext,amsthm}
\usepackage{latexsym, amssymb}
\usepackage{graphicx}
\usepackage{subfigure}
\usepackage{pgf,tikz}
\usepackage{pstricks}
 \usepackage{float}
 \usepackage{longtable}
\newtheorem{thm}{Theorem}[section]

\newtheorem{lem}[thm]{Lemma}
\newtheorem{prop}[thm]{Proposition}
\theoremstyle{definition}
\newtheorem{defn}[thm]{Definition}
\theoremstyle{remark}

\numberwithin{equation}{section}
\begin{document}

\title{\bfseries\textrm{An Entropy-based Measure of Intelligence Degree of System Structures*}
\footnotetext{*The idea of this paper originates partly from the author's unpublished study under the guidance of professor Lei Guo. This research is supported by the Fundamental Research Funds for
the Central Universities under Grant 2021RC267.\\
\,\,Wei Su is with School of Mathematics and Statistics, Beijing Jiaotong University, Beijing 100044, China, {\tt su.wei@bjtu.edu.cn}.
}
 }
\author{Wei Su}%
\date{}%
\maketitle
\begin{abstract}
In this paper, we investigate how to measure the intelligence of systems under specific structures. Two indicators are adopted to characterize the intelligence of a given structure, namely the function diversity of the structure, and the ability to generate order under specific environments. A measure of intelligence degree is proposed, with which the intelligence degree of several basic structures is calculated. It is shown that some structures are indeed ``smarter'' than the others under the proposed measure. The results add a possible way of revealing the evolution mechanism of natural life and constructing life-like structures with high intelligence degree.
\end{abstract}
\section{Introduction}
Who is smarter between a rock and an E. coli? The answer seems obvious, E. coli is smarter than rocks. However, this conclusion comes from intuition. To answer this question rigorously, we need to be able to quantitatively measure the intelligence of an object. This paper mainly studies the measurement of system intelligence.

The intelligent behavior of a system can be attributed to the intelligence of its functions. In this paper, when we say that a system has intelligent functions, we mean that the system has two characteristics. First, the function of the system has possibilities of evolution, which first means that the system has a high functional reserve. For different states, the system can produce different macro outputs (that is, functions). Second, the system is capable of generating purposeful order; that is, in a specific environment, the system can achieve specific functions.

Considering the correspondence between function and structure, the above two characteristics of the intelligent function of a system can be transformed into the same two characteristics of the system structure. The question is how can the intelligence of a given structure be measured? Corresponding to the above two characteristics, we will consider two indicators of the intelligence of system structure: first, the function diversity of a system structure, which means the evolution possibility of system function; second, the ability of system structure to generate order. To measure the intelligence of a given structure, we need to calculate the above two indicators. The first indicator, we define the concept of \emph{functional entropy}, which is used to measure the diversity of the reserve functions of a system structure, that is, the evolutionary potential of system functions; for the second indicator, namely the ability of attaining order under specific environment, we consider the thermal motion of particles as targeted environment since this is the most fundamental environment of life in nature. And we measure the ability of a system to generate order through the \emph{Boltzmann entropy} by adding the system structure to a closed particle system with thermal motion. Of course, we can also measure the ability of a system structure to attain different orders by putting them in different environments, while in this paper we only consider the environment of random thermal motion.

In the following, we first introduce the two types of entropy, then give the formula of intelligence degree of a given structure, and finally measure the intelligence of several basic structures.

\section{Definitions of two types of entropy}
%

\subsection{Functional entropy}
In this part, we define functional entropy to measure the diversity of macro outputs produced by system for different initial states.
Consider a dynamical system $(X,\mathcal{U},\mu,F,\mathbb{Z}^+)$, where $(X,\mathcal{U},\mu)$ is a probability space, $\mathcal{U}$ is $\sigma$-algebra, $\mu$ is a probability measure, $F:X\rightarrow X$ is a measurable transformation of $X$ which represents the structure of the evolutionary system, and $\mathbb{Z}^+$ is the set of positive integer  for $F$ iterating on $X$. Let $\mathcal{P}$ be a finite measurable  partition of $X$, and $ F^{-1}(\mathcal{P})=\{A|A=F^{-1}(B),B\in\mathcal{P}\}, F^{-1}(B)=\{x\in X|F(x)\in B\}$, then $F^{-1}(\mathcal{P})$ is also a finite measurable partition of $X$.
\begin{defn}\label{Def:funcentropybasic}
Denote $F^{-t}(\mathcal{P})=F^{-1}(F^{-(t-1)}(\mathcal{P})), t\in \mathbb{Z}^+ $, then
\begin{equation}\label{Equa:funcentroy}
  h_\mu(F)=\sup\limits_{\mathcal{P}}\lim\limits_{t\rightarrow\infty}-\sum\limits_{A\in F^{-t}(\mathcal{P})}\mu(A)\ln\mu(A)
\end{equation}
is called functional entropy of the structure $F$ under measure $\mu$.
\end{defn}

The larger the functional entropy $h_\mu(F)$ is, the larger the output diversity of $F$ is. Compared to topological entropy, $X$ is not required to be compact, and compared to measure entropy, $F$ is not required to preserve the measure. Hence Definition \ref{Def:funcentropybasic} of functional entropy can be applied to a wider range of dynamical systems.

For a given structure, functional entropy may be different under different measure $\mu$. At the same time, according to actual situation, if the macro structure of partition is determined in advance, a specific type of functional entropy can also be defined. Here we introduce two special types of functional entropy.

\begin{defn}\label{Def:funcentropyinfspace}
Let $\mu$ be a probability measure equivalent to Lebesgue measure, $D_r=\{x\in \mathbb{R}^n| |x_i| \leq \frac{r}{2}\}$ be a n-dimensional cube of side length $r>0$ in $\mathbb{R}^n$, and $\mathcal{P}^r_C=\mathbb{R}^n-D_r$. Given $\epsilon>0$, consider a net partition $\mathcal{P}^\epsilon_{D_r}$ of $D_r$ with side length $\epsilon$, then $\mathcal{P}^r=\mathcal{P}^\epsilon_{D_r}\bigcup \mathcal{P}^r_C$ is a finite partition of $\mathbb{R}^n$. Define
\begin{equation}\label{Equa:standentroyepsi}
  h_\epsilon(F)=\lim\limits_{r\rightarrow\infty}\lim\limits_{t\rightarrow\infty}-\frac{1}{\ln|\mathcal{P}^r|}\sum\limits_{A\in F^{-t}(\mathcal{P}^r)}\mu(A)\ln\mu(A),
\end{equation}
then $h_\epsilon(F)$ is called $\epsilon$-standard functional entropy of the structure $F$. Here, $|\cdot|$ denotes the cardinality of a set.
\end{defn}
In Definition \ref{Def:funcentropyinfspace}, the state space is the whole space. Since $\mu$ is a probability measure, it is not uniform in the whole space, so it is not suitable to compare the functional entropy of different structures. The following is a case where the state space is bounded and closed.
\begin{defn}\label{Def:funcentropyboundspace}
Let $\mu$ be a Lebesgue measure on the unit n-cube $D_1$, and $\mathcal{P}_\epsilon$ be a uniform net partition with a given side length $\epsilon>0$ of $D_1$. Define
\begin{equation}\label{Equa:standentroy}
  h(F)=\lim\limits_{\epsilon\rightarrow 0}\lim\limits_{t\rightarrow\infty}-\frac{1}{\ln|\mathcal{P}_\epsilon|}\sum\limits_{A\in F^{-t}(\mathcal{P}_\epsilon)}\mu(A)\ln\mu(A),
\end{equation}
then $h(F)$ is called standard functional entropy of the structure $F$.
\end{defn}
Standard functional entropy $h(F)$ is based on a given probability measure with a uniform distribution function and uniform partition, which can be used to compare the functional entropy of different structures.
Since the systems in reality are mainly considered in bounded real space, we mainly use the standard functional entropy defined in
(\ref{Equa:standentroy}) to study the intelligence degree of a given structure in this paper.

For the standard functional entropy defined in (\ref{Equa:standentroy}), we have
\begin{thm}\label{Thm:funcentropy}
Let $\mu$ be a probability measure, then $h(F)\in[0,1]$.
\end{thm}
First we need the following lemma:
\begin{lem}\label{Lem:avermaxi}\cite{Tomasz2011}
Let $P_m=(p_1,p_2,\ldots,p_m), m\geq 1$ where $p_i\geq0 (i=1,\ldots, m), \sum_{i=1}^mp_i=1$, and denote the entropy function $H(P_m)=-\sum_{i=1}^m p_i\ln p_i$, then $\sup_{P_m}H(P_m)=\ln m$.
\end{lem}
\begin{proof}[Proof of Theorem \ref{Thm:funcentropy}]
Given any $\epsilon>0, t>0$, and the partition $\mathcal{P}_\epsilon$, we have
\begin{equation}\label{Equa:FnPrleqnPr}
  |F^{-t}(\mathcal{P}_\epsilon)|\leq |\mathcal{P}_\epsilon|.
\end{equation}
Then by Lemma \ref{Lem:avermaxi},(\ref{Equa:standentroy}) and (\ref{Equa:FnPrleqnPr})
\begin{equation*}
\begin{split}
  h(F)=&\lim\limits_{\epsilon\rightarrow 0}\lim\limits_{t\rightarrow\infty}-\frac{1}{\ln|\mathcal{P}_\epsilon|}\sum\limits_{A\in F^{-t}(\mathcal{P}_\epsilon)}\mu(A)\ln\mu(A)\\
  \leq &\lim\limits_{\epsilon\rightarrow 0}\lim\limits_{t\rightarrow\infty}\frac{1}{\ln|\mathcal{P}_\epsilon|}\ln |F^{-t}(\mathcal{P}_\epsilon)|\\
  \leq&\lim\limits_{\epsilon\rightarrow 0}\lim\limits_{t\rightarrow\infty}\frac{1}{\ln|\mathcal{P}_\epsilon|}\ln |\mathcal{P}_\epsilon|\\
  =&1.
\end{split}
\end{equation*}
This completes the proof.
\end{proof}
\begin{thm}\label{Thm:boundedFentropy}
 For any structure $F$, if $\{ \lim\limits_{t\rightarrow\infty}F^t(x)| x\in  D_1\}$ is a  set consisting of finite points in $ D_1$, then its standard functional entropy $h(F)=0$.
\end{thm}
\begin{proof}
Let $\mathbb{K}=\{ \lim\limits_{t\rightarrow\infty}F^t(x)| x\in  D_1\}$,
then there exists $M<\infty$ such that $|\mathbb{K}|=M$ and for any $S\subset  D_1-\mathbb{K}$, we have
\begin{equation}\label{Equa:Fnegtempty}
  \lim\limits_{t\rightarrow\infty}F^{-t}(S)=\emptyset.
\end{equation}
Given $\epsilon>0$ and the uniform net partition $\mathcal{P}_\epsilon$, by (\ref{Equa:Fnegtempty}),
\begin{equation}\label{Equa:Fnegtprbr0}
  \lim\limits_{t\rightarrow\infty}|F^{-t}(\mathcal{P}_\epsilon)|\leq \lim\limits_{t\rightarrow\infty}|F^{-t}(\mathbb{K})|\leq M.
\end{equation}
Notice $|\mathcal{P}_\epsilon|\rightarrow \infty$, as $\epsilon\rightarrow 0$, then by Lemma \ref{Lem:avermaxi},(\ref{Equa:standentroy}) and (\ref{Equa:Fnegtprbr0}), we have
\begin{equation*}
\begin{split}
  h(F)=&\lim\limits_{\epsilon\rightarrow 0}\lim\limits_{t\rightarrow\infty}-\frac{1}{\ln|\mathcal{P}_\epsilon|}\sum\limits_{A\in F^{-t}(\mathcal{P}_\epsilon)}\mu(A)\ln\mu(A)\\
  \leq& \lim\limits_{\epsilon\rightarrow 0}\lim\limits_{t\rightarrow\infty}\frac{1}{\ln|\mathcal{P}_\epsilon|}\ln |F^{-t}(\mathcal{P}_\epsilon)|\\
  \leq& \lim\limits_{\epsilon\rightarrow 0}\frac{1}{\ln|\mathcal{P}_\epsilon|}\ln M\\
  =&0.
\end{split}
\end{equation*}
This completes the proof.
\end{proof}
Next there are standard functional entropy of several systems with specific structures.
\begin{prop}\label{Prop:canstantproj}
Given $x_0\in  D_1$, if $F(x)\equiv x_0, x\in  D_1$, then $h(F)=0$.
\end{prop}
\begin{prop}\label{Prop:indenticalproj}
If $F(x)\equiv x, x\in  D_1$, then $h(F)=1$.
\end{prop}
\begin{prop}\label{Prop:feedback}
Consider the system with feedback structure
\begin{equation}\label{Model:feedback}
\dot{x}(t)=Ax(t),\qquad \,x(0)\in  D_1.
\end{equation}
When all eigenvalues of $A$ satisfy $Re(\lambda_i(A))<0$, the system (\ref{Model:feedback}) is called pure negative feedback structure, and denote as $F_-$;
When all $Re(\lambda_i(A))>0$, the system is called pure positive feedback structure, and denote as $F_+$.
Then $h(F_-)=0$,$h(F_+)=0$.
\end{prop}
\begin{proof}
For pure negative feedback structure $F_-$, it is obvious that $\lim\limits_{t\rightarrow\infty}F_-^{t}(x)=0, x\in  D_1$. By Theorem \ref{Thm:boundedFentropy}, we have $h(F_-)=0$. For the pure positive feedback structure $F_+$, by the theory of differential equations, we know
\begin{equation}\label{Equa:posifeedsolu}
\begin{split}
    \lim\limits_{t\rightarrow\infty}F_+^{t}(x)=& 0,\qquad x=0 \\
    \parallel\lim\limits_{t\rightarrow\infty}F_+^{t}(x)\parallel=& \infty, \qquad x\neq 0.
\end{split}
\end{equation}
Given $\epsilon>0$ and the uniform net partition $\mathcal{P}_\epsilon$, from (\ref{Equa:posifeedsolu})
\begin{equation}\label{Equa:posifeedfnegt}
  \lim\limits_{t\rightarrow\infty}F^{-t}(\mathcal{P}_\epsilon)=\{0\},\qquad 
\end{equation}
Since $\mu$ is Lebesgue measure, then
\begin{equation}\label{Equa:posifeedfnegtprob}
 \mu(\{0\})=0.
\end{equation}
By (\ref{Equa:standentroy}), (\ref{Equa:posifeedfnegt}) and (\ref{Equa:posifeedfnegtprob}), we have
\begin{equation*}
\begin{split}
  h(F)=&\lim\limits_{\epsilon\rightarrow 0}\lim\limits_{t\rightarrow\infty}-\frac{1}{\ln|\mathcal{P}_\epsilon|}\sum\limits_{A\in F^{-t}(\mathcal{P}_\epsilon)}\mu(A)\ln\mu(A)\\
  =& \lim\limits_{\epsilon\rightarrow 0}-\frac{1}{\ln|\mathcal{P}_\epsilon|}\mu(\{0\})\ln\mu(\{0\})\\
  =&0.
\end{split}
\end{equation*}
This completes the proof.
\end{proof}
\begin{prop}\label{Prop:entropylinear}
Consider the discrete linear system
\begin{equation}\label{Model:dislinearmodel}
  x(t+1)=Ax(t),\quad x(0)\in  D_1,\quad t\geq 0,
\end{equation}
and let $\lambda_i(A)$ be the eigenvalues of $A$. If all $|\lambda_i(A)|< 1$ or $|\lambda_i(A)|>1$, then the standard functional entropy $h(F)=0$.
\end{prop}

\subsection{Boltzmann entropy}
We will use Boltzmann entropy to measure the ability of a structure $F$ to generate order in a environment of random thermal motion. Boltzmann entropy is a quantity that describes the number of microscopic states corresponding to the macroscopic states of a closed thermal system, as defined below
\begin{equation}\label{Equa:Bolentropy}
  S=k\ln \Omega.
\end{equation}
Here $S$ is Boltzmann entropy, $k$ is Boltzmann constant, and $\Omega$ is the number of microscopic states in the equilibrium.

Intuitively, the higher for Boltzmann entropy, the more disordered for system. The second law of thermodynamics states that a closed thermal system of particles without interaction eventually approaches a state of maximum entropy.
If a specific interaction is added among particles, that is, a specific structure is given to the system, it will generate certain kind of order. So the Boltzmann entropy in the equilibrium state of new system decreases, that is, the microscopic number of macroscopic states is reduced, and the reduced amount can be used to measure the ability of the structure to generate order.

\section{Measure of Intelligence}
\subsection{Intelligence degree}
For a system, its intelligence can be measured by the functional entropy of its structure $F$ and the Boltzmann entropy in equilibrium state. Intuitively, the higher the functional entropy $h(F)$ is, the greater the macro output diversity of the structure is. The lower the Boltzmann entropy $B(F)$ in equilibrium state, the stronger the ability of system to generate order. An intelligent structure should have high macro output diversities and strong order-generating ability.
\begin{defn}\label{Def:intellmeasure}
Let $B_0$ be the maximum Boltzmann entropy of a given number of particles without interaction in a closed thermal system, and denote
\begin{equation}
 D(F)=h(F)\Big(1-\frac{B(F)}{B_0}\Big),
\end{equation}
then $D(F)$ can be regarded as the intelligence degree of structure $F$.
\end{defn}
Hence, if a structure $F$ possesses higher functional entropy $h(F)$ as well as lower Boltzmann entropy $B(F)$, it owns higher intelligence degree. It seems high functional entropy is contradictory to low Boltzmann entropy, however, later in the paper, we will show a structure with both higher functional entropy and low Boltzmann entropy.

\subsection{Intelligence degree of several typical structures}
Using the formula of intelligence degree defined above, we will examine the intelligence degree of several typical structures, where the standard functional entropy is chosen in (\ref{Equa:standentroy}).
\subsubsection{Closed thermal particle system}

Closed thermal particle system is a fundamental and classical physical system and also a basic research object of thermodynamics. An ideal system of particles with thermal motion can be considered as a random walk system with no interaction among particles. The second law of thermodynamics states that the system will eventually approach a state of maximum Boltzmann entropy $B_0$, i.e. $B(F)=B_0$, a weakest order-generating ability. Meanwhile, Proposition \ref{Prop:indenticalproj} shows that closed thermal particle system has the maximum functional entropy of 1, a strongest ability of function diversity.
According to the definition of intelligence degree (\ref{Def:intellmeasure}), the intelligence degree of closed thermal particle system is 0.
\subsubsection{Solid structure}

The immobile individuals in a system of solid structure can be described by a constant mapping $F_I(x)\equiv x_0, x\in  D_1$ where $x_0$ is a fixed point in $ D_1$. According to Propostion \ref{Prop:canstantproj}, the functional entropy of constant mapping is $h(F_I)=0$. And the Boltzmann entropy satisfies
$B (F_I) = 0$. Thus by (\ref{Def:intellmeasure}), we know $D(F_I)=0$. As the other extreme structure compared to closed thermal particle system, it indicates that solid structure has a strongest ability of generating order, but no possibility of function evolution, so its intelligence is 0. This also indicates that a stone has an intelligence degree of 0 under the definition (\ref{Def:intellmeasure}).

\subsubsection{The multi-agent system with fixed topological structures}
Consider a typical multi-agent system consisting of $n$ individuals
\begin{equation}\label{Model:multiagentmodel}
  x(t+1)=Ax(t),\quad x(0)\in  D_1,\quad t\geq 0,
\end{equation}
where $A\in\mathbb{R}^{n\times n}$ is a random matrix, i.e., all elements of $A$ are non-negative and the row sum is 1. $A$ represents the interaction between agents, that is, the topology of the system. $A$ can be either time-varying or time-invariant. Such topological relations are determined in advance, so we call such system structures fixed topological structures. Here, we just consider of $A$ which is time invariant.

For the standard functional entropy of the system (\ref{Model:multiagentmodel}), we consider two different structures of $A$. First, $A$ is irreducible, implying the topology of the system is strongly connected; Second, $A$ is reducible, implying that the topology of the system is not strongly connected.

\begin{lem}\label{Lem:stoindumatrix}
Consider the system (\ref{Model:multiagentmodel}) and denote $F_1(x)=Ax, x\in D_1$. If the stochastic matrix $A$ is irreducible, then the standard functional entropy of the system
\begin{equation}\label{Equa:topoentrmas}
  h(F_1)\in\Big(0,\frac{1}{n}\Big].
\end{equation}
\end{lem}
\begin{proof}
Since $A$ is irreducible, then the topological graph of the system is strongly connected, and from \cite{Jad2003}, the system will reach consensus, namely
\begin{equation*}
  \lim_{t\rightarrow\infty}x(t)=c(x(0))\mathbf{1}.
\end{equation*}
Here, $c(x(0))$ is a constant vector depending on $x(0)$. Given $\epsilon>0$ and the uniform partition $\mathcal{P}_\epsilon$ in (\ref{Equa:standentroy}), let $L=\{x\in D_1|x_1=x_2=\ldots=x_n\}$ and denote
\begin{equation}\label{Equa:LcutPepsi}
  L\bigwedge \mathcal{P}_\epsilon=\{L\bigcap A|A\in\mathcal{P}_\epsilon\},
\end{equation}
it is not hard to show that
\begin{equation}\label{Equa:LPnonP}
  |L\bigwedge \mathcal{P}_\epsilon|=|\mathcal{P}_\epsilon|^{1/n}.
\end{equation}
Meanwhile, for any $t>0$ and the partition $\mathcal{P}$, we have
\begin{equation}\label{Equa:FnPumleqnP}
  |F_1^{-t}(\mathcal{P})|\leq |\mathcal{P}|.
\end{equation}
By Lemma \ref{Lem:avermaxi}, (\ref{Equa:LPnonP}), (\ref{Equa:FnPumleqnP}) and Definition \ref{Equa:standentroy}, we obtain
\begin{equation*}
\begin{split}
 h(F_1)= &\lim\limits_{\epsilon\rightarrow 0}\lim\limits_{t\rightarrow\infty}-\frac{1}{\ln|\mathcal{P}_\epsilon|}\sum\limits_{A\in F_1^{-t}(\mathcal{P}_\epsilon)}\mu(A)\ln\mu(A)\\
 \leq&\lim\limits_{\epsilon\rightarrow 0}\lim\limits_{t\rightarrow\infty}\frac{1}{\ln|\mathcal{P}_\epsilon|}\ln|F_1^{-t}(L\bigwedge \mathcal{P}_\epsilon)|\\
 \leq&\lim\limits_{\epsilon\rightarrow 0}\frac{1}{\ln|\mathcal{P}_\epsilon|}\ln|L\bigwedge \mathcal{P}_\epsilon|\\
 = &\frac{1}{n}.
 \end{split}
\end{equation*}
This completes the proof.
\end{proof}
Lemma \ref{Lem:stoindumatrix}
shows that the standard functional entropy of the structures with strongly connected fixed interaction does not exceed $\frac{1}{n}$ ($n$ is the agent number of system). So, the larger amount of population in a system with fixed connected structure, the lower its functional entropy may be. In other words, large population size is not beneficial to function evolution of a system with fixed connected interactions.

When $A$ is reducible, the topological graph of the structure $F$ is not strongly connected, and denoted as $F_1'$. Since $A$ is a stochastic matrix, we know that the standard functional entropy
\begin{equation}\label{Equa:linearfunentro}
  h(F_1)\leq h(F_1').
\end{equation}
This shows that the function diversity of strongly connected structure is not higher than that of unconnected topologies.

For the Boltzmann entropy of the system (\ref{Model:multiagentmodel}) in equilibrium, we also need consider the two connected structures of $A$.

Let $A$ act on a thermal system consisting of $n$ particles, and the system equation is
\begin{equation}\label{Model:noisymamodel}
  x(t+1)=Ax(t)+\xi(t+1),\quad x(t)\in \mathbb{R}^n,\quad t\geq 0,
\end{equation}
where $\xi(t), t\geq 1$ are i.i.d. random variables, and $\parallel\xi(t)\parallel\leq \delta, a.s., t\geq 0$, $\delta$ is arbitrary small positive real number.

Let the Boltzmann entropy in equilibrium of the system (\ref{Model:noisymamodel}) with a strongly connected topology be $B_1$, the corresponding number of microstates be $\Omega_1$, and the Boltzmann entropy of the non-strongly connected topology be $B_1'$, the corresponding number of microstates be $\Omega_1'$. By \cite{Wang2007}, the system with strongly connected topology will reach robust consensus, while the non-strongly connected structure cannot. So $\Omega_1< \Omega_1'\leq \Omega$ ($\Omega$ is the number of microstates in equilibrium for the same thermal particles without interaction), and therefore
\begin{equation}\label{Equa:linearBe}
  B_1< B_1'.
\end{equation}
This shows that the structure with strongly connected topology possesses stronger ability of generating order than the structure of non-strongly connected topology.

Let the intelligence degree of the strongly connected structure is $D_1$, and the intelligence degree of the non-strongly connected structure is $D_1'$, then from the definition of intelligence degree (\ref{Def:intellmeasure}) and (\ref{Equa:topoentrmas}),
\begin{equation}\label{Equa:linearId}
  D_1=h(F_1)\Big(1-\frac{B_1}{B_0}\Big)\in\Big(0,\frac{1}{n}\Big),\quad D_1'=h(F_1')\Big(1-\frac{B_1'}{B_0}\Big).
\end{equation}

\subsubsection{The self-organizing system with state-dependent local interaction}

In this section, we consider a self-organizing system whose structure is determined by state-dependent local interaction,
\begin{equation}\label{Model:HKmodel}
  x_i(t+1)=\frac{1}{|\mathcal{N}_i(t)|}\sum\limits_{j\in \mathcal{N}_i(t)}x_j(t),\quad i\in\mathcal{V}=\{1,\ldots,n\},\quad x(0)\in D_1
\end{equation}
where
\begin{equation}\label{Model:neigh}
 \mathcal{N}_i(t)=\{j\in\mathcal{V}\; \big|\; |x_j(t)-x_i(t)|\leq \rho\}
\end{equation}
is the neighbor set of agent $i$, $\rho\in[0,1]$ is the neighbor radius of agents.
The system (\ref{Model:HKmodel}) is called Hegselmann-Krause model \cite{Hegselmann2002} and, like the structure (\ref{Model:multiagentmodel}), is a structure based on averaging mechanism. This state-dependent locally interacting structure captures a self-organizing mechanism which exists ubiquitously in real complex systems, biological and neural networks for instance, and is included in many models, such as the classic Vicsek model. But quite different from the structure (\ref{Model:multiagentmodel}), the topology of system (\ref{Model:HKmodel}) is not predetermined though time-varying, and it heavily depends on the current state of system.

Now let's consider the standard functional entropy of system (\ref{Model:HKmodel}), denoted as $h(F_{HK})$. For system (\ref{Model:HKmodel}), the following conclusion holds:
\begin{prop}\label{Prop:HKfrag}
For every $1\leq i\leq n$ in the system (\ref{Model:HKmodel}), $x_i(t)$ will convergent to the static state $x_i^*$, and for each $i, j$, it holds $x_i^*=x_j^*$ or $|x_i^*-x_j^*|> \rho$.
\end{prop}
Obviously, when $\max_{i,j\in\mathcal{V}}|x_i(0)-x_j(0)|\leq \rho$, we have $x_1^*=x_2^*=\ldots=x_n^*=\frac{1}{n}\sum_{i\in\mathcal{V}} x_i(0)$, and when $\min_{i,j\in\mathcal{V}}|x_i(0)-x_j(0)|> \rho$, we have $x_i^*=x_i(0),i\in\mathcal{V}$.
\begin{prop}\label{Prop:funentrHK}
Given neighbor radius $\rho$, denote $h(F_{HK_\rho})$ as the standard functional entropy of the system (\ref{Model:HKmodel}), then \begin{equation*}
  h(F_{HK_1})=\frac{1}{n},\quad h(F_{HK_0})=1,
\end{equation*}
and
\begin{equation*}
  \frac{1}{n}< h(F_{HK_\rho})\leq 1, \qquad 0<\rho<1.
\end{equation*}
\end{prop}
\begin{proof}
When $\rho=0$, there is no interaction between the particles of the system, then by Proposition \ref{Prop:indenticalproj}, we know $h(F_{HK_0})=1$. By equations of (\ref{Model:HKmodel}) and Proposition
\ref{Prop:HKfrag}, when $\rho= 1$, we have $x_i^*=\frac{x_1(0)+\ldots+x_n(0)}{n}, i\in\mathcal{V}$. Since $\mu$ is Lebesgue measure, then by the proof of Lemma \ref{Lem:stoindumatrix}, we obtain $h(F_{HK_1})=\frac{1}{n}$.

Now consider the case of $0<\rho_0<1$, and arbitrarily given $\epsilon>0$ and uniform net partition $\mathcal{P}_\epsilon$.
By properties of model (\ref{Model:HKmodel}), we can prove that, for two groups of initial states $x(0), y(0)$, if $\sum x_i(0)\neq \sum_i y(0)$, then $x^*\neq y^*$. Hence, for sufficiently small $\epsilon$, say, when $\epsilon\rightarrow 0$, there exists $A'\in L\bigvee \mathcal{P}_\epsilon$(refer to $(\ref{Equa:LcutPepsi})$), such that $\lim\limits_{t\rightarrow\infty}F_{HK_{\rho_0}}^{-t}(A)\subset \lim\limits_{t\rightarrow\infty}F_{HK_1}^{-t}(A')$ for any $A\in\mathcal{P}_\epsilon$. Moreover, $|F_{HK_{\rho_0}}^{-t}(\mathcal{P}_\epsilon)|>|F_{HK_{1}}^{-t}(\mathcal{P}_\epsilon)|$,
then by the subadditivity of entropy function \cite{Tomasz2011}, we have
\begin{equation*}
  -\sum\limits_{A\in F_{HK_{\rho_0}}^{-t}(\mathcal{P}_\epsilon)}\mu(A)\ln\mu(A)> -\sum\limits_{A'\in F_{HK_{1}}^{-t}(\mathcal{P}_\epsilon)}\mu(A')\ln\mu(A')
\end{equation*}
it then follows
\begin{equation*}
\begin{split}
 h(F_{HK_{\rho_0}})= &\lim\limits_{\epsilon\rightarrow 0}\lim\limits_{t\rightarrow\infty}-\frac{1}{\ln|\mathcal{P}_\epsilon|}\sum\limits_{A\in F_{HK_{\rho_0}}^{-t}(\mathcal{P}_\epsilon)}\mu(A)\ln\mu(A)\\
 > &\lim\limits_{\epsilon\rightarrow 0}\lim\limits_{t\rightarrow\infty}-\frac{1}{\ln|\mathcal{P}_\epsilon|}\sum\limits_{A'\in F_{HK_{1}}^{-t}(\mathcal{P}_\epsilon)}\mu(A')\ln\mu(A')\\
 =& h(F_{HK_1})=\frac{1}{n}.
 \end{split}
\end{equation*}
This completes the proof.
\end{proof}
Now consider the Boltzmann entropy of system (\ref{Model:HKmodel}) in equilibrium state. Let $F_{HK_\rho}$ act on a closed thermal system, and the system equation is
\begin{equation}\label{Model:HKmodelnoise}
  x_i(t+1)=\frac{1}{|\mathcal{N}_i(t)|}\sum\limits_{j\in \mathcal{N}_i(t)}x_j(t)+\xi_i(t+1), \quad i\in\mathcal{V}=\{1,\ldots,n\}
\end{equation}
where $|\xi_i(t)|\leq \delta\leq\rho/2, a.s., i\in\mathcal{V}, t\geq 1$, $\delta$ is an arbitrarily small positive real number and $0<\rho<1$.
Denote Boltzmann entropy of the system (\ref{Model:HKmodelnoise}) as $B_2$. From \cite{Su2017auto}, we know the system (\ref{Model:HKmodel}) will reach quasi-consensus. This indicates that the structure $F_{HK}$ possesses a same ability of generating order with the structure of strongly connected fixed topology $F_1$, i.e., $B_{HK_\rho}=B_1$. Meanwhile, by Proposition \ref{Prop:funentrHK}, we have $h(F_{HK_\rho})>\frac{1}{n}\geq h(F_1)$, then by Definition \ref{Def:intellmeasure}, we know the intelligence degree of system (\ref{Model:HKmodel}) satisfies
\begin{equation}\label{Equa:HKId}
  D(F_{HK_\rho})=h(F_{HK_\rho})\Big(1-\frac{B_{HK_\rho}}{B_0}\Big)>h(F_1)\Big(1-\frac{B_1}{B_0}\Big) =D(F_1).
\end{equation}
This implies that the intelligence degree of the structure with state-dependent local interaction is higher than that with strongly connected fixed interactions. Roughly speaking, $F_1$ possesses strong interaction among agents, and hence strong ability of generating order but weak function evolvability; while $F_{HK}$ possesses weaker interaction among agents, and hence higher function evolvability but parallel strong ability of generating order. Though the structures $F_1$ and $F_{HK_\rho}$ appears in existing multi-agent models with no special distinction, the results here show that there exists an essential difference between them.

\subsubsection{Feedback structure}

Consider the following system with pure feedback structure.
\begin{equation}\label{Model:negfeed}
\dot{x}(t)=Ax(t),\qquad \,x(0)\in  D_1
\end{equation}

If all the eigenvalues of $A$ satisfy $Re(\lambda_i(A))<0$, the system (\ref{Model:negfeed}) is called pure negative feedback structure, and denoted by $F_-$. By Proposition \ref{Prop:feedback}, we know the standard functional entropy $h(F_-)=0$. Next we check the Boltzmann entropy of the system in equilibrium state. Consider the following system with random drive,
\begin{equation}\label{Model:negfeedrand}
dy_t=Ay_t+\delta d D_t,\qquad \,x(0)\in  D_1
\end{equation}
where $B_t, t\geq 0$ is Brownian motion. By the theory of stochastic differential equations, we know the solution $y(t)$ of the system (\ref{Model:negfeedrand}) and the solution $x(t)$ of the system (\ref{Model:negfeed}) satisfy
\begin{equation*}
  \lim\limits_{\delta\rightarrow 0}\mathbb{E}(\max_{t\leq T}|y(t)-x(t)|)^2=0, \qquad \forall T>0.
\end{equation*}
Since $\lim_{t\rightarrow \infty}x(t)=0$, then when $\delta$ is sufficiently small, the Boltzmann entropy of the system in equilibrium state is also quite small. It means the structure of pure negative feedback possesses very strong ability to generate order. However, as the standard functional entropy of the system is 0, the system function lack evolvability. Then by Definition \ref{Def:intellmeasure}, we know the intelligence degree of the system (\ref{Model:negfeed}) with pure negative feedback structure is 0.

When all $Re(\lambda_i(a))>0$, the system is pure positive feedback structure, and denoted by $F_+$. Similarly, by Proposition \ref{Prop:feedback}, we know the standard functional entropy $h(F_+)=0$. By Definition \ref{Def:intellmeasure}, the intelligence degree of the system (\ref{Model:negfeed}) with pure positive feedback structure is 0.

The above conclusions show that the intelligence degree of both pure positive and pure negative feedback structures is 0. However, for the system composed of positive and negative feedback, such as Lorentz chaotic system, the standard functional entropy is positive. In other words, by adding positive feedback to negative feedback structure, the standard function entropy of system can be increased and the function diversity of system can be improved. But for chaotic system, its Boltzmann entropy is also very high, so its intelligence may be very weak. But it still indicates that to achieve intelligent systems, a combination of positive and negative feedback structures is necessary. And one can choose different positive and negative feedback combination modes, to improve the intelligence of system.

At last of the section, we summarize the intelligence degree of the above basic structures
\begin{table}[H]
\centering
\caption{Intelligence degree of the basic structures}
\scalebox{0.75}{
\begin{tabular}{|c|c|c|c|}\hline
Structures (F)&Functional entropy (h(F))& Boltzmann entropy (B(F)) & Intelligence degree (D(F)) \\\hline
No interaction & 1 & $B_0$ &0\\\hline
Solid & 0 & 0 &0\\\hline
Connected fixed topology & $h_1\in(0,\frac{1}{n}]$ & $B_1$ & $h_1(1-\frac{B_1}{B_0})$\\\hline
Self-organizing with local interaction & $h_2\in (\frac{1}{n},1]$ & $B_1$ & $h_2(1-\frac{B_1}{B_0})$\\\hline
Pure negative feedback & 0 & $B_N$ &0\\\hline
Pure positive feedback & 0 & $B_P$ &0\\\hline
\end{tabular}}
\end{table}
\section{Conclusions}

This paper discusses the measurement of intelligence of system structures, defines function entropy to measure function diversity of system, and uses standard function entropy and Boltzmann entropy to define the intelligence degree of a given structure. Then, the intelligence degree of several basic structures is examined. Under the definition, the intelligence degree of closed thermal particle system, solid structure and pure positive or negative feedback structure is 0, while the intelligence degree of the structure based on local interaction is higher than that of fixed strongly connected topology. At the same time, neither positive feedback nor negative feedback alone can produce intelligent structure, which implies only a structure combined with positive and negative feedbacks can produce intelligent behavior.

Since natural life is a self-organizing process of different materials under different interacting structures, the results in this paper provide possibilities of mimicking the evolution process of natural intelligent life, especially biological intelligence, through self-organizing integrations of different basic structures under specific environments. Also, we can expect to explore diverse intelligent constructions finally based on the results in this paper.

%
%
%
%
%
%
%
%

\end{document}